\documentclass[a4paper,11pt]{article}
\usepackage {cite}
\usepackage[utf8]{inputenc}
\usepackage {amsmath,amssymb}
\usepackage {graphicx}
\usepackage [usenames,dvipsnames] {color}
\usepackage {times}
\usepackage {verbatim}
\usepackage {paralist,enumerate,enumitem}
\usepackage {tabularx,booktabs,multirow}
\usepackage{hyperref}

\usepackage{todonotes}

\usepackage[margin=2.56cm]{geometry}
\usepackage{amsthm}
\usepackage{authblk}

\usepackage[left,pagewise] {lineno}
\definecolor {infocolor} {rgb} {0.6,0.6,0.6}

\newtheorem{theorem}{Theorem}

\newtheorem{lemma}{Lemma}

\newtheorem{proposition}{Proposition}
\newenvironment {claim}
{\noindent {\em Claim.}\ } {\smallskip}

\DeclareGraphicsExtensions{.pdf,.png,.eps}
\graphicspath{{figs/}}

\newcommand {\etal} {\textit {et al.}}

\newtheorem {observation}[theorem]{Observation}
\newenvironment {claimproof}{\noindent \textit{Proof of Claim:}}{\hfill \ensuremath {\triangle}\medskip}

\makeatletter
\DeclareRobustCommand{\bfseries}{%
   \not@math@alphabet\bfseries\mathbf
   \fontseries\bfdefault\selectfont
   \boldmath
}
\makeatother

\DeclareMathOperator{\suc}{succ}
\DeclareMathOperator{\outdeg}{out-deg}

\newcommand{\calR}{\mathcal{R}}

\title {Combinatorial Properties of Triangle-Free Rectangle Arrangements and the Squarability Problem}

  \author[1,3]{Jonathan Klawitter}
  \author[2]{Martin N\"ollenburg}
  \author[3]{Torsten Ueckerdt}
  \affil[1]{Institut f\"{u}r Theoretische Informatik, Karlsruhe Institute of Technology, Germany}
  \affil[2]{Algorithms and Complexity Group, TU Wien, Vienna, Austria}
  \affil[3]{Institut f\"{u}r Algebra und Geometrie, Karlsruhe Institute of Technology, Germany}
  
  \date{}

\newcounter{dummycount}

\newcommand{\wormholeLem}[1]
{
\newcounter{#1}
\setcounter{#1}{\value{lemma}}
}
\newcommand{\wormholeThm}[1]
{
\newcounter{#1}
\setcounter{#1}{\value{theorem}}
}

\newcommand{\wormholeProp}[1]
{
\newcounter{#1}
\setcounter{#1}{\value{proposition}}
}

\newenvironment{backInTimeLem}[1]
{
\setcounter{dummycount}{\value{lemma}}
\setcounter{lemma}{\value{#1}}
}
{
\setcounter{lemma}{\value{dummycount}}
}

\newenvironment{backInTimeThm}[1]
{
\setcounter{dummycount}{\value{theorem}}
\setcounter{theorem}{\value{#1}}
}
{
\setcounter{theorem}{\value{dummycount}}
}

\newenvironment{backInTimeProp}[1]
{
\setcounter{dummycount}{\value{proposition}}
\setcounter{proposition}{\value{#1}}
}
{
\setcounter{proposition}{\value{dummycount}}
}

\begin{document}

\maketitle

\begin{abstract}
 We consider %
 arrangements of axis-aligned rectangles in the plane.
 A geometric arrangement specifies the coordinates of all rectangles, while a combinatorial arrangement specifies only the respective intersection type in which each pair of rectangles intersects. %
 First, we investigate combinatorial contact arrangements, i.e., arrangements of interior-disjoint rectangles, with a triangle-free intersection graph.
 We show that such rectangle arrangements are in bijection with the $4$-orientations of an underlying planar multigraph and prove that there is a corresponding geometric rectangle contact arrangement.
 Moreover, we prove that every triangle-free planar graph is the contact graph of such an arrangement.
 Secondly, we introduce the question whether a given rectangle arrangement has a combinatorially equivalent square arrangement. In addition to some necessary conditions and counterexamples, we show that rectangle arrangements pierced by a horizontal line are squarable under certain sufficient conditions.  
\end{abstract}

\section{Introduction}

We consider arrangements of axis-aligned rectangles and squares in the plane.
Besides \emph{geometric rectangle arrangements}, in which all rectangles are given with coordinates, we are also interested in \emph{combinatorial rectangle arrangements}, i.e., equivalence classes of combinatorially equivalent arrangements.
Our contribution is two-fold.

First we consider maximal (with a maximal number of contacts) combinatorial rectangle contact arrangements, in which no three rectangles share a point.
For rectangle arrangements this is equivalent to the contact graph being \emph{triangle-free}, unlike, e.g., for triangle contact arrangements.
We prove a series of analogues to the well-known maximal combinatorial triangle contact arrangements and to Schnyder realizers.
The contact graph $G$ of a maximal triangle contact arrangement is a maximal planar graph.
A \emph{$3$-orientation} is an orientation of the edges of a graph $G'$, obtained from $G$ by adding six edges (two at each outer vertex), in which every vertex has exactly three outgoing edges.
Each outer vertex has two outgoing edges that end in the outer face without having an endpoint there.
A \emph{Schnyder realizer}~\cite{s-epgg-90,schnyderposet89} is a $3$-orientation of $G'$ together with a coloring of its edges with colors $0,1,2$ such that every vertex has exactly one outgoing edge in each color and incoming edges are colored in the color of the ``opposite'' outgoing edge.
The three outgoing edges represent the three corners of a triangle and the color specifies the corner, see Fig.~\ref{fig:Schnyder}.
De Fraysseix~\etal~\cite{FraysseixTContact} proved that the maximal combinatorial triangle contact arrangements of $G$ are in bijection with the $3$-orientations of $G'$ and the Schnyder realizers of $G'$.
Schnyder proved that for every maximal planar graph $G$, $G'$ admits a Schnyder realizer and hence $G$ is a triangle contact graph.

\begin{figure}[tb]
 \centering
 \includegraphics{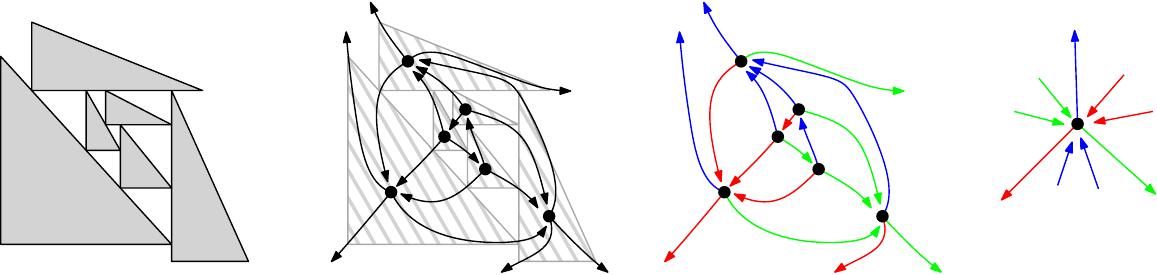}
 \caption{Left to right: Maximal combinatorial contact arrangement with axis-aligned triangles, no three sharing a point. $3$-orientation of $G'$. Schnyder realizer of $G'$. Local coloring rules for Schnyder realizer.}
 \label{fig:Schnyder}
\end{figure}

In this paper we prove an analogous result, which, roughly speaking, is the following.
We consider maximal triangle-free combinatorial rectangle contact  arrangements.
The corresponding contact graph $G$ is planar with all faces of length $4$ or $5$.
We define an underlying plane multigraph $\bar{G}$, whose vertex set also includes a vertex for each inner face of the contact graph, and define \emph{$4$-orientations} of $\bar{G}$.
Here, every vertex has exactly four outgoing edges, where each outer vertex has two edges ending in the outer face.
For a $4$-orientation we introduce \emph{corner-edge-labelings} of $\bar{G}$, which are, similar to Schnyder realizers, colorings of the outgoing edges at vertices of $\bar{G}$ corresponding to rectangles with colors $0,1,2,3$ satisfying certain local rules.
Each outgoing edge represents a corner of a rectangle and the color specifies which corner it is, see Fig.~\ref{fig:our-structures}.
We then prove that the combinatorial contact arrangements of $G$ are in bijection with the $4$-orientations of $\bar{G}$ and the corner-edge-labelings of $\bar{G}$.
We also prove that for every maximal triangle-free planar graph $G$, $\bar{G}$ admits a $4$-orientation and hence that $G$ is a  rectangle contact graph.

\begin{figure}[tb]
 \centering
 \includegraphics{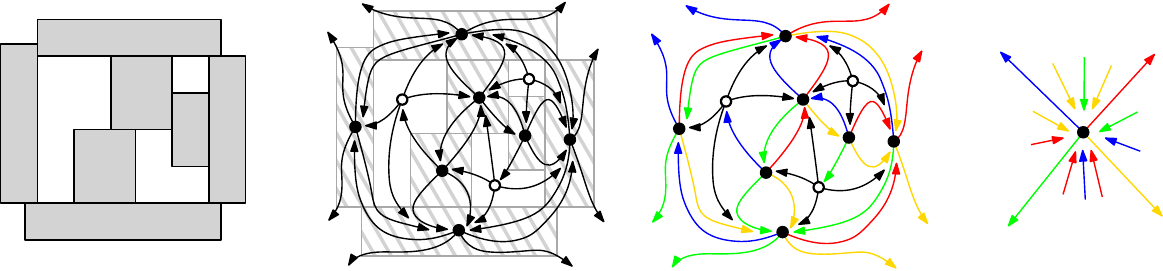}
 \caption{Left to right: Maximal combinatorial contact arrangement with axis-aligned rectangles, no three sharing a point. $4$-orientation of underlying graph. Corner-edge-labeling of underlying graph. Local coloring rules for corner-edge-labeling.}
 \label{fig:our-structures}
\end{figure}

Our second result is concerned with the question whether a given geometric rectangle arrangement can be transformed into a combinatorially equivalent square arrangement.
The similar question whether a pseudocircle arrangement can be transformed into a combinatorially equivalent circle arrangement has recently been studied by Kang and Müller~\cite{km-apc-14}, who showed  (among other results) that the problem is NP-hard.
We say that a rectangle arrangement can be \emph{squared} (or is \emph{squarable}) if an equivalent square arrangement exists.
Obviously, squares are a very restricted class of rectangles and  hence it does not come as a surprise that 
not every rectangle arrangement can be squared.
The natural open question is to characterize the squarable rectangle arrangements and to answer the complexity status of the corresponding decision problem.
As a first step towards solving these questions, we show, on the one hand, some general necessary conditions  for squarability  and, on the other hand, sufficient conditions implying that certain subclasses of rectangle arrangements are always squarable.

\paragraph{Related Work.}

Intersection graphs and contact graphs of axis-aligned rectangles or squares in the plane are a popular, almost classic, topic in discrete mathematics and theoretical computer science with lots of applications in computational geometry, graph drawing and VLSI chip design.
Most of the research for rectangle intersection graphs concerns their recognition~\cite{Yan83}, colorability~\cite{AG60} or the design of efficient algorithms such as for finding maximum cliques~\cite{IA83}.
On the other hand, rectangle contact graphs are mainly investigated for their combinatorial and structural properties.
Almost all the research here concerns edge-maximal $3$-connected rectangle contact graphs, so called \emph{rectangular duals}.
These can be characterized by the absence of separating triangles~\cite{kk-rdpg-85,Ung53} and the corresponding representations by touching rectangles can be seen as dissections of a rectangle into rectangles.
Combinatorially equivalent dissections are in bijection with regular edge labelings~\cite{KanHe97} and transversal structures~\cite{Fus09}.
The question whether a rectangular dual has a rectangle dissection in which all rectangles are squares has been investigated by Felsner~\cite{FelsnerSurvey}.

\section{Preliminaries}

In this paper a \emph{rectangle} is an axis-aligned rectangle in the plane, i.e., the cross product $[x_1,x_2] \times [y_1,y_2]$ of two bounded closed intervals.
A \emph{geometric rectangle arrangement} is a finite set $\calR$ of rectangles; it is a \emph{contact arrangement} if any two rectangles have disjoint interiors.
In a contact arrangement, any two non-disjoint rectangles $R_1, R_2$ have one of the two contact types \emph{side contact} and \emph{corner contact}, see Fig.~\ref{fig:intersection-types} (left); we exclude the degenerate case of two rectangles sharing only one point.
If $\calR$ is not a contact arrangement, four intersection types are possible: \emph{side piercing}, \emph{corner intersection}, \emph{crossing}, and \emph{containment}, see Fig.~\ref{fig:intersection-types} (right).
Note that side contact and corner contact are degenerate cases of side piercing and corner intersection, whereas crossing and containment have no analogues in contact arrangements.
If no two rectangles form a crossing, we say that $\calR$ is \emph{cross-free}.
Moreover, in each type (except containment) it is further distinguished which sides of the rectangles touch or intersect.

\begin{figure}[tb]
 \centering
 \includegraphics[page=2]{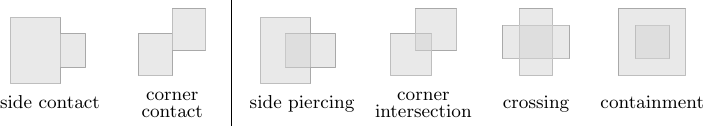}
 \caption{Contact types (left) and intersection types (right) of rectangles.}
 \label{fig:intersection-types}
\end{figure}

Two rectangle arrangements $\calR_1$ and $\calR_2$ are \emph{combinatorially equivalent} if $\calR_1$ can be continuously deformed into $\calR_2$ such that every intermediate state is a rectangle arrangement with the same intersection or contact type for every pair of rectangles. %
An equivalence class of combinatorially equivalent arrangements is called a \emph{combinatorial rectangle arrangement}.
So while a geometric arrangement specifies the coordinates of all rectangles, think of a combinatorial arrangement as specifying only the way in which any two rectangles touch or intersect.
In particular, a combinatorial rectangle arrangement is defined by
\textbf{(1)} for each rectangle $R$ and each side of $R$ the counterclockwise order of all intersecting (touching) rectangle edges, labeled by their rectangle $R'$ and the respective side of $R'$ (top, bottom, left, right),
\textbf{(2)} for containments the respective component of the arrangement, in which a rectangle is contained. %

In the \emph{intersection graph} of a rectangle arrangement there is one vertex for each rectangle and two vertices are adjacent if and only if the corresponding rectangles intersect.
As combinatorially equivalent arrangements have the same intersection graph,  combinatorial arrangements themselves have a well-defined intersection graph.
For  rectangle contact arrangements (combinatorial or geometric) the intersection graph is also called the \emph{contact graph}.
Note that such contact graphs are planar, as we excluded the case of four rectangles meeting in a corner.

\section{Statement of Results}

\subsection{Maximal triangle-free planar graphs and rectangle contact arrangements}
\label{subsec:triangle-free}

We consider so-called \emph{MTP-graphs}, that is, (M)aximal (T)riangle-free (P)lane graphs with a quadrangular outer face.
Note that each face in such an MTP-graph is a $4$-cycle or $5$-cycle, and that every plane triangle-free graph is an induced subgraph of some MTP-graph.
Given an MTP-graph $G$ a rectangle contact arrangement of $G$ is one whose contact graph is $G$, where the embedding inherited from the arrangement is the given embedding of $G$, and where each outer rectangle has two corners in the unbounded region\footnote{Other configurations of the outer four rectangles can be easily derived from this.}.
We define the closure, $4$-orientations and corner-edge-labelings:

\medskip

 \begin{compactdesc}
  \item[The \emph{closure}] $\bar{G}$ of $G$ is derived from $G$ by replacing each edge of $G$ with a pair of parallel edges, called an \emph{edge pair}, and adding into each inner face $f$ of $G$ a new vertex, also denoted by $f$, connected by an edge, called a \emph{loose edge}, to each vertex incident to that face.
  At each outer vertex we add two loose edges pointing into the outer face, although we do not add a vertex for the outer face.
  Note that $\bar{G}$ inherits a unique plane embedding with each inner face being a triangle or a $2$-gon.

  \item[A $4$-orientation] of $\bar{G}$ is an orientation of the edges and half-edges of $\bar{G}$ such that every vertex has outdegree exactly $4$.
  An edge pair is called \emph{uni-directed} if it is oriented consistently and \emph{bi-directed} otherwise.

  \item[A corner-edge-labeling] of $\bar{G}$ is a $4$-orientation of $\bar{G}$ together with a coloring of the outgoing edges of $\bar{G}$ at each vertex of $G$ with colors $0,1,2,3$ (see Fig.~\ref{fig:local-colors}) such that
   \begin{compactenum}[(i)]
    \item around each vertex $v$ of $G$ we have outgoing edges in color $0,1,2,3$ in this counterclockwise order and\label{enum:outgoing-edges}
    \item in the wedge, called \emph{incoming wedge}, at $v$ counterclockwise between the outgoing edges of color $i$ and $i+1$ there are some (possibly none) incoming edges colored $i+2$ or $i+3$, $i=0,1,2,3$, all indices modulo $4$.\label{enum:incoming-wedges}
   \end{compactenum}
 \end{compactdesc} 

\medskip
 
In a corner-edge-labeling the four outgoing edges at a vertex of $\bar{G}$ corresponding to a face of $G$ are not colored.
Further we remark that~\eqref{enum:outgoing-edges} implies that uni-directed pairs are colored $i$ and $i-1$, while~\eqref{enum:incoming-wedges} implies that bi-directed pairs are colored $i$ and $i+2$, for some $i \in \{0,1,2,3\}$, where all indices are considered modulo $4$.

The following theorem is proved in Sec.~\ref{sec:triangle-free}.

\begin{figure}[tb]
 \centering
 \includegraphics[scale=1]{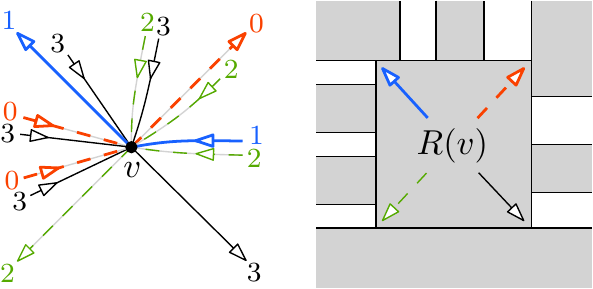} 
 \caption{Local color patterns in corner-edge-labelings of an MTP-graph at a vertex $v$, together with the corresponding part in a rectangle contact arrangement.}
 \label{fig:local-colors}
\end{figure}

\begin{theorem}\label{thm:triangle-free}
 Let $G$ be an MTP-graph, then each of the following are in bijection:
 \begin{compactitem}
  \item the combinatorial  rectangle  contact arrangements of $G$
  \item the corner-edge-labelings of $\bar{G}$
  \item the $4$-orientations of $\bar{G}$
 \end{compactitem}
\end{theorem}

Using the bijection between $4$-orientations of $\bar{G}$ and combinatorial rectangle contact arrangements of $G$ given in Thm.~\ref{thm:triangle-free}, we can show that every MTP-graph $G$ is a  rectangle contact graph, which is the statement of the next theorem; its proof is found in Sec.~\ref{sec:MPT-has-orientation}. 

\wormholeThm{thm:MPT-has-orientation}
\begin{theorem}\label{thm:MPT-has-orientation}
 Every MTP-graph has a  rectangle contact arrangement and it can be computed in linear time.
\end{theorem}

We remark that our technique in the proof of Thm.~\ref{thm:triangle-free} constructs in linear time a geometric rectangle contact arrangement in the $2n \times 2n$ square grid, where $n$ is the number of vertices in the graph.
Thus also the rectangle contact arrangement in Thm.~\ref{thm:MPT-has-orientation} uses only on linear-size grid.

\subsection{Squarability and line-pierced rectangle arrangements}
\label{subsec:squarability}

In the squarability problem, we are given a rectangle arrangement $\calR$ and want to decide whether $\calR$ can be squared, i.e., whether it has a combinatorially equivalent square arrangement.
The first observation is that there are obvious obstructions to the squarability of a rectangle arrangement. If any two rectangles in $\calR$ are crossing (see Fig.~\ref{fig:intersection-types}) then there are obviously no two combinatorially equivalent squares.

But even if we restrict ourselves to cross-free rectangle arrangements, we can find unsquarable configurations.
One such arrangement is depicted in Fig.~\ref{fig:unsquarable}~(left).
To get an unsquarable arrangement with a triangle-free intersection graph, we can use the fact that two side-piercing rectangles translate immediately into a smaller-than relation for the corresponding squares: the side length of the square to pierce into the side of another square needs to be strictly smaller.
Hence any rectangle arrangement that contains a cycle of side-piercing rectangles cannot be squarable, see Fig.~\ref{fig:unsquarable}~(middle).
Moreover, we may even create a counterexample of a rectangle arrangement whose intersection graph is a path and that causes a geometrically infeasible configuration for squares, see Fig.~\ref{fig:unsquarable}~(right). %

\begin{figure}[tb]
 \centering
 \includegraphics[scale=0.8]{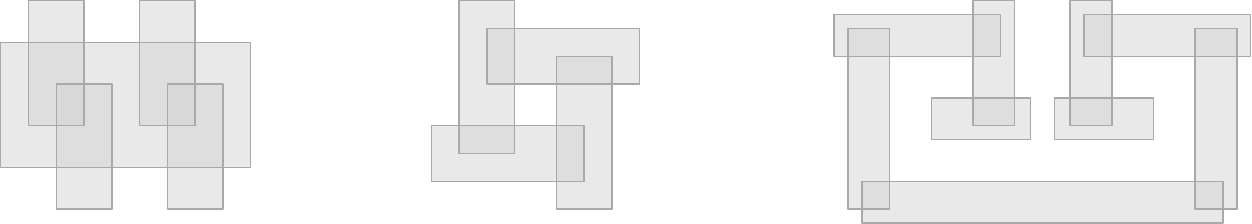}
 \caption{Three cross-free unsquarable rectangle arrangements.}
 \label{fig:unsquarable}
\end{figure}

\begin{proposition}\label{pro:unsquarable}
 Some cross-free rectangle arrangements are unsquarable, even if the intersection graph is a path.
\end{proposition}

Therefore we focus on a non-trivial subclass of rectangle arrangements that we call line-pierced.
A rectangle arrangement $\calR$ is \emph{line-pierced} if there exists a horizontal line $\ell$ such that $\ell \cap R \neq \emptyset$ for all $R \in \calR$.
The line-piercing strongly restricts the possible vertical positions of the rectangles in $\calR$, which lets us prove two sufficient conditions for squarability in the following theorem.

\begin{theorem}\label{thm:line-pierced}
        Let $\calR$ be a cross-free, line-pierced rectangle arrangement.
        \begin{compactitem}
                \item If $\calR$ is triangle-free, then $\calR$ is squarable.
                \item If $\calR$ has only corner intersections, then $\calR$ is squarable, even using line-pierced unit squares.
        \end{compactitem}
\end{theorem}

On the other hand, cross-free, line-pierced rectangle arrangements in general may have forbidden cycles or other geometric obstructions to squarability. We give two examples in Sec.~\ref{sec:line-pierced}, together with  the proof  of Thm.~\ref{thm:line-pierced}. %

\section{Bijections between 4-Orientations, Corner-Edge-Labelings and Rectangle Contact Arrangements -- Proof of Theorem~\ref{thm:triangle-free}}
\label{sec:triangle-free}

Throughout this section let $G = (V,E)$ be a fixed MTP-graph and $\bar{G}$ be its closure.
By definition, every corner-edge-labeling of $\bar{G}$ induces a $4$-orientation of $\bar{G}$.
We prove Thm.~\ref{thm:triangle-free}, i.e., that combinatorial rectangle contact arrangements of $G$, $4$-orientations of $\bar{G}$ and corner-edge-labelings of $\bar{G}$ are in bijection, in three steps:

\medskip

\begin{compactitem}
 \item Every rectangle contact arrangement of $G$ induces a $4$-orientation of $\bar{G}$. (Lemma~\ref{lem:rectangles-to-orientation})

 \item Every $4$-orientation of $\bar{G}$ induces a corner-edge-labeling of $\bar{G}$. (Lemma~\ref{lem:orientation-to-labeling})

 \item Every corner-edge-labeling of $\bar{G}$ induces a rectangle contact arrangement of $G$. (Lemma~\ref{lem:labeling-to-rectangles})
\end{compactitem}

\subsection{From rectangle arrangements to $4$-orientations.}

\wormholeLem{lem:rectangles-to-orientation}
\begin{lemma}\label{lem:rectangles-to-orientation}
 Every rectangle contact arrangement of $G$ induces a $4$-orientation of $\bar{G}$.
\end{lemma}

\begin{proof}
 Consider any rectangle contact arrangement $\{R(v) \mid v \in V\}$ of $G$, we define a $4$-orientation of $\bar{G}$ as follows.
 Consider any vertex $v$ in $G$, the rectangle $R(v)$ and the set $C(v)$ of points on the boundary of $R(v)$ that are corners of some rectangle in the arrangement.
 We draw vertex $v$ inside $R(v)$ and draw a straight segment to each point in $C(v)$.
 See Fig.~\ref{fig:rectangles-to-orientation}.
 Such a segment is supposed to be half an edge in $\bar{G}$ incident to $v$.
 We orient a segment outgoing at $v$ if the corresponding corner in $C(v)$ is a corner of $R(v)$, and incoming at $v$ otherwise.

 \begin{figure}[tb]
  \centering
  \includegraphics{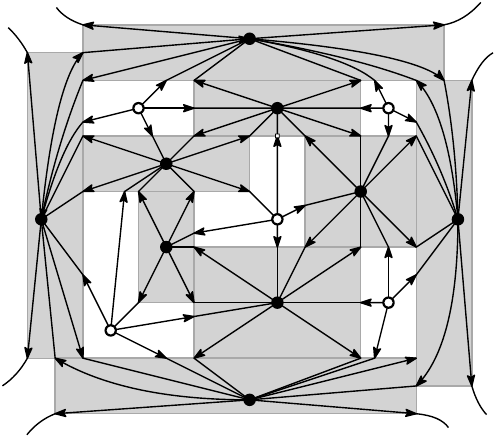}
  \caption{Obtaining a $4$-orientation of $\bar{G}$ from a rectangle contact arrangement of $G$.}
  \label{fig:rectangles-to-orientation}
 \end{figure}

 Next, consider any inner face $f$ in $G$ and the corresponding bounded component of $\mathbb{R}^2 - \bigcup_{v \in V} R(v)$, which is an axis-aligned polygon denoted by $P(f)$.
 In fact, $P(f)$ is a $4$-gon or $6$-gon when $f$ is a $4$-face or $5$-face, respectively.
 We draw the vertex $f$ in $\bar{G}$ inside $P(f)$ and draw outgoing edges from $f$ to the four vertices of $G$ whose rectangles constitute the four extremal (topmost, rightmost, bottommost and leftmost) sides of $P(f)$.
 If $P(f)$ is a $6$-gon we additionally draw a half-edge from its concave corner to $f$, oriented incoming at $f$.
 See again Fig.~\ref{fig:rectangles-to-orientation} for an illustration.
\end{proof}

We continue with a crucial property of $4$-orientations.
For a simple cycle $C$ of $G$, consider the corresponding cycle $\bar{C}$ of edge pairs in $\bar{G}$.
The \emph{interior} of $\bar{C}$ is the bounded component of $\mathbb{R}^2$ incident to all vertices in $C$ after the removal of all vertices and edges of $\bar{C}$.
In a fixed $4$-orientation of $\bar{G}$ a directed edge $e=(u,v)$ \emph{points inside} $C$ if $u \in V(C)$ and $e$ lies in the interior of $\bar{C}$, i.e., either $v$ lies in the interior of $C$, or $e$ is a chord of $\bar{C}$ in the interior of $\bar{C}$.

\wormholeLem{lem:pendant-at-cycle}
\begin{lemma}\label{lem:pendant-at-cycle}
 For every $4$-orientation of $\bar{G}$ and every simple cycle $C$ of $G$ the number of edges pointing inside $C$ is exactly $|V(C)|-4$.
\end{lemma}

\begin{proof}
 We prove the statement by induction on the number of inner faces in the interior of $C$.
 If $C$ is the boundary of a single inner face $f$, then exactly $|V(C)|-4$ edges are incoming at $f$ (as $\outdeg(f) = 4$), which are exactly the edges pointing inside $C$.
 Otherwise, there is a path $P$ in the interior of $C$ that shares only its endpoints with $C$.
 Path $P$ splits $C$ into two simple cycles $C_1$ and $C_2$ having exactly the edges of $P$ in common, each of which has fewer inner faces of $G$ in its interior.

 By induction hypothesis $C_i$ has exactly $|E(C_i)|-4$ edges pointing inside, for $i=1,2$.
 Let $a$, $a \in \{0,\ldots,4\}$, be the number of edges in $\bar{G}$ that lie in $P$ and are outgoing at an endpoint of $P$.
 The remaining $2|E(P)|-a$ edges in $P$ are outgoing at some vertex in $P \setminus C$.
 As there are exactly $|E(P)|-1$ such vertices, each with out-degree $4$, we conclude that exactly $4(|E(P)|-1)-(2|E(P)|-a) = 2|E(P)|-4+a$ edges that point from a vertex of $P \setminus C$ inside $C_1$ or $C_2$.
 Plugging things together, the number of edges pointing inside $C$ is exactly
 \begin{align*}
 |E(C_1)|-4 + |E(C_2)|-4 + a - (2|E(P)|-4+a) &= (|E(C_1)|+|E(C_2)|-2|E(P)|) - 4\\
  &= |E(C)| - 4,
 \end{align*}
 as desired.
\end{proof}

\subsection{From $4$-orientations to corner-edge-labelings.}

Next we shall show how a $4$-orientation of $\bar{G}$ can be augmented (by choosing colors for the edges) into a corner-edge-labeling.
Fix a $4$-orientation.
If $e$ is a directed edge in an edge pair, then $e$ is called a \emph{left edge}, respectively \emph{right edge}, when the $2$-gon enclosed by the edge pair lies on the right, respectively on the left, when going along $e$ in its direction.
Thus, a uni-directed edge pair consists of one left edge and one right edge, while a bi-directed edge pair either consists of two left edges (clockwise oriented $2$-gon) or two right edges (counterclockwise oriented $2$-gon).

If $e = (u,v)$ is an edge in an edge pair, let $e_2$ and $e_3$ be the second and third outgoing edge at $v$ when going counterclockwise around $v$ starting with $e$.
We define the \emph{successor} of $e$ as $\suc(e) = e_2$ if $e$ is a right edge, and $\suc(e) = e_3$ if $e$ is a left edge, see Fig.~\ref{fig:successor}~(b,c).
Note that in a corner-edge-labeling $\suc(e)$ is exactly the outgoing edge at $v$ that has the same color as $e$, see Fig.~\ref{fig:local-colors}.

\begin{figure}[tb]
 \centering
 \includegraphics[width=\textwidth]{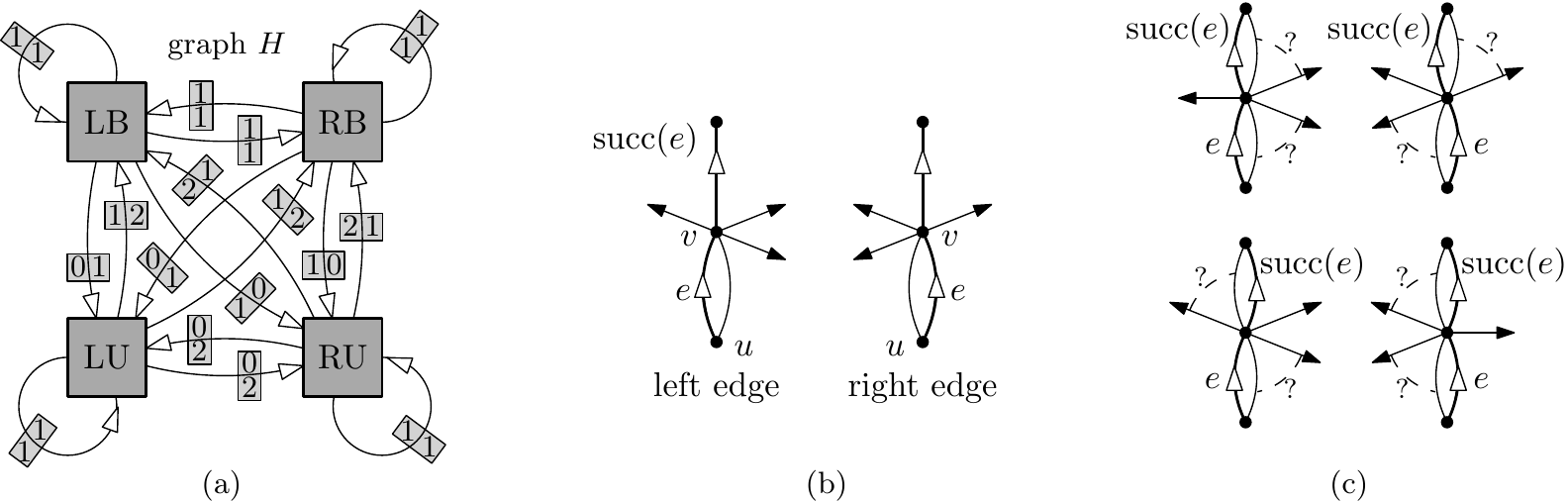}
 \caption{(a) The graph $H$. \textbf{L}, \textbf{R}, \textbf{U}, \textbf{B} stands for left edge, right edge, uni-directed and bi-directed edge pair, respectively. The number of outgoing edges in the left and right wedge are shown on the left and right of the corresponding arrow.
 (b) Illustration of the definition of $\suc(e)$.
 (c) Summarizing the $16$ possible cases for $e$ and $\suc(e)$. Edges connected by a dashed arc may or may not coincide.}
 \label{fig:successor}
\end{figure}

Note that $e' = \suc(e)$ may be a loose edge in $\bar{G}$ at the concave vertex for some $5$-face in $G$.
For the sake of shorter proofs below, we shall avoid the treatment of this case.
To do so, we augment $G$ to a supergraph $G'$ such that starting with any edge in any edge pair and repeatedly taking the successor, we never run into a loose edge pointing to an inner face.

The graph $G'$ is formally obtained from $G$ by stacking a new vertex $w$ into each $5$-face $f$, with an edge to the incoming neighbor $v$ of $f$ in $\bar{G}$ and the vertex $u$ at $f$ that comes second after $v$ in the clockwise order around $f$ in $\bar{G}$.
(Indeed, the second vertex in counterclockwise order would be equally good for our purposes.)
Let $f_1$ and $f_2$ be the resulting $4$-face and $5$-face incident to $w$, respectively.
We obtain a $4$-orientation of the closure $\bar{G'}$ of $G'$ by orienting all edges at $f_1$ as outgoing, both edges between $v$ and $w$ as right edges (counterclockwise), the remaining three edges at $w$ as outgoing, and the remaining four edges at $f_2$ as outgoing.
See Fig.~\ref{fig:force-last-edges}~(left) for an illustration.

\begin{figure}[tb]
 \centering
 \includegraphics[scale=0.8]{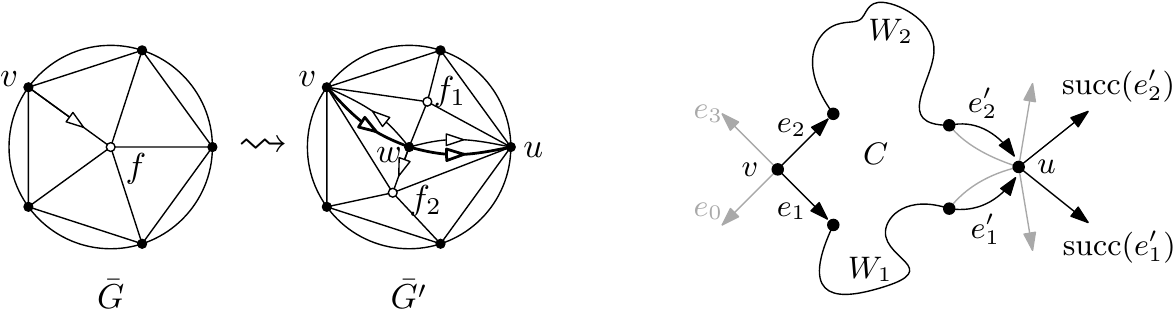}
 \caption{Left: Stacking a new vertex $w$ into a $5$-face $f$ of $G$. The orientation of edges on the boundary of $f$, as well as outgoing edges at $f$, $f_1$, $f_2$ is omitted. The directed edge $(v,w)$ and its successor $(w,u)$ are highlighted. Right: Illustration of the proof of the Claim in the proof of Lemma~\ref{lem:orientation-to-labeling}.}
 \label{fig:force-last-edges}
\end{figure}

Before we augment the $4$-orientation of $\bar{G'}$ into a corner-edge-labeling, we need one last observation.
Let $e$ and $\suc(e)$ be two edges in edge pairs of $\bar{G'}$ with common vertex $v$.
Consider the wedges at $v$ between $e$ and $\suc(e)$ when going clockwise (left wedge) and counterclockwise (right wedge) around $v$.
Each of $e$, $\suc(e)$ can be a left edge or right edge, and in a uni-directed pair or a bi-directed pair.
This gives us four types of edges and $16$ possibilities for the types of $e$ and $\suc(e)$.
The graph $H$ in Fig.~\ref{fig:successor}~(a) shows for each of these $16$ possibilities the number of outgoing edges at $v$ in the left and right wedge at $v$.

\begin{observation}\label{obs:count-left-right}
 For every directed closed walk on $k$ edges in the graph $H$ in Fig.~\ref{fig:successor}~(a) we have
 \[
  \#\text{edges in left wedges} = \#\text{edges in right wedges} = k.
 \]
\end{observation}
\begin{proof}
 It suffices to check each directed cycle on $k$ edges, $k=1,2,3,4$.
\end{proof}

\begin{lemma}\label{lem:orientation-to-labeling}
 Every $4$-orientation of $\bar{G}$ induces a corner-edge-labeling of $\bar{G}$.
\end{lemma}
\begin{proof}
 Consider the augmented graph $G'$, its closure $\bar{G'}$ and $4$-orientation as defined above.
 For any edge $e$ in an edge pair in $\bar{G'}$ (and hence every edge of $\bar{G}$ outgoing at some vertex of $G$) consider the directed walk $W_e$ in $\bar{G'}$ starting with $e$ by repeatedly taking the successor as long as it exists (namely the current edge is in an edge pair).
 
 First we show that $W_e$ is a simple path ending at one of the eight loose edges in the outer face.
 Indeed, otherwise $W_e$ would contain a simple cycle $C$ where every edge on $C$, except the first, is the successor of its preceding edge on $C$.
 From the graph $H$ of Fig.~\ref{fig:successor}~(a) we see that every wedge of $C$ contains at most two outgoing edges.
 With Obs.~\ref{obs:count-left-right} the number of edges pointing inside $C$ is at least $|V(C)|-2$ and at most $|V(C)|+2$, which is a contradiction to Lemma~\ref{lem:pendant-at-cycle}.
 
 Now let $v_0,v_1,v_2,v_3$ be the outer vertices in this counterclockwise order.
 Define the color of $e$ to be $i$ if $W_e$ ends with the right loose edge at $v_i$ or the left loose edge at $v_{i-1}$, indices modulo $4$.
 By definition every edge has the same color as its successor in $\bar{G'}$ (if it exists).
 Thus this coloring is a corner-edge-labeling of $\bar{G'}$ if at every vertex $v$ of $G$ the four outgoing edges are colored $0$, $1$, $2$, and $3$, in this counterclockwise order around $v$.
 
 \begin{claim}%
  Let $e_1,e_2$ be two outgoing edges at $v$ for which $W_{e_1} \cap W_{e_2}$ consists of more than just $v$.
  Then $e_1$ and $e_2$ appear consecutively among the outgoing edges around $v$, say $e_1$ clockwise after $e_2$.

  Moreover, if $u \neq v$ is a vertex in $W_{e_1}\cap W_{e_2}$ for which the subpaths $W_1$ of $W_{e_1}$ and $W_2$ of $W_{e_2}$ between $v$ and $u$ do not share inner vertices, then the last edge $e'_1$ of $W_1$ is a right edge and the last edge $e'_2$ of $W_2$ is a left edge, $e'_1$ and $e'_2$ are part of (possibly the same) uni-directed pairs and these pairs sit in the same incoming wedge at $u$.
 \end{claim}

 \begin{claimproof}
  If $W_{e_1} \cap W_{e_2}$ consists of more than just $v$, then clearly there is such a vertex $u$ for which the subpaths $W_1$ of $W_{e_1}$ and $W_2$ of $W_{e_2}$ between $v$ and $u$ do not share inner vertices.
  Consider the simple cycle $C$ in $\bar{G'}$ formed by $W_1$ and $W_2$.
  Assume without loss of generality that going along $W_1$ the interior of $C$ lies on the left, see Fig.~\ref{fig:force-last-edges}~(right).
  Note that if $u$ is incoming endpoint of $e_1$ and $e_2$, then the statement clearly holds.
  Thus we may assume that $|V(C)| > 2$.

  If $k_1 = |E(W_1)|$ and $k_2 = |E(W_2)|$, then by Obs.~\ref{obs:count-left-right} the number of edges pointing inside $C$ from inner vertices of $W_1$ and $W_2$ is at least $(k_1-2) + (k_2-2)$ and at most $(k_1 + 2) + (k_2+2)$.
  By Lemma~\ref{lem:pendant-at-cycle} the total number of edges pointing inside $C$ is $|V(C)|-4 = k_1+k_2-4$.
  It follows there are exactly $k_1-2$ such edges at inner vertices of $W_1$, exactly $k_2 -2$ such edges at inner vertices of $W_2$ and no such edge at $u$ or $v$.
  From the graph $H$ in Fig.~\ref{fig:successor}~(a)
  we can deduce that $e'_1$ is the right edge in a uni-directed pair, as desired, because \textbf{RU} is the only vertex of $H$ with outgoing arrows labeled $2$ on the left.
  Moreover, $e_1$ can not be the right edge in a uni-directed pair, too, because the edge in $H$ from \textbf{RU} to \textbf{RU} is labeled only $1$ on the left.
  Similarly, $e'_2$ must be the right edge in a uni-directed pair as desired and $e_2$ not the right edge in a uni-directed pair.
  This finally implies that $e_1$ and $e_2$ are consecutive in the counterclockwise order around $v$ and hence concludes to proof of the claim.
 \end{claimproof}

 The claim implies that the two walks $W_{e_1}$ and $W_{e_2}$ can neither cross, nor have an edge in common.
 Considering the four walks starting in a given vertex, we can argue (with the second part of the claim) that our coloring is a corner-edge-labeling of $\bar{G'}$.
 Finally, we inherit a corner-edge-labeling of $\bar{G}$ by reverting the stacking of artificial vertices in $5$-faces.
\end{proof}

\subsection{From corner-edge-labelings to rectangle contact arrangements.}

We shall show how a rectangle arrangement of $G$ can be constructed from a corner-edge-labeling of $\bar{G}$.
Fix a corner-edge-labeling of $\bar{G}$.
We shall construct a rectangle contact arrangement $\{R(v) \mid v \in V\}$ of $G$ with $R(v) = [x_1(v),x_2(v)] \times [y_1(v),y_2(v)]$, which is compatible with the given corner-edge-labeling.
For every vertex $v$ of $G$ we have two pairs of variables $x_1(v),x_2(v)$ and $y_1(v),y_2(v)$ corresponding to the $x$-coordinates of the vertical sides of $R(v)$ ($x_1$ for left and $x_2$ for right) and the $y$-coordinates of the $y$-coordinates of the horizontal sides of $R(v)$ ($y_1$ for lower and $y_2$ for upper).
Now $R(v)$ is a rectangle if and only if
\begin{equation}
 x_1(v) < x_2(v) \qquad \text{ and } \qquad y_1(v) < y_2(v). \label{eq:proper-rectangle}
\end{equation}

For every edge $vw$ of $G$ the way in which $R(v)$ and $R(w)$ are supposed to touch (which two corners lie in $R(v) \cap R(w)$) is encoded in the given corner-edge-labeling.
This can be described by more equations and inequalities in terms of variables $x_1(v)$, $x_2(v)$, $x_1(w)$, $x_2(w)$, $y_1(v)$, $y_2(v)$, $y_1(w)$, $y_2(w)$, which depend only on the coloring and orientation of the edge pair between $v$ and $w$ in $\bar{G}$.
For example, if the edge pair is uni-directed, outgoing at $v$, incoming at $w$ and colored with $0$ and $1$, then the top side of $R(v)$ is supposed to be contained in the bottom side of $R(w)$, which is the case if and only if $y_2(v) = y_1(w)$ and $x_1(w) < x_1(v) < x_2(v) < x_2(w)$.
The complete system of inequalities and equalities is given in Table~\ref{tab:constraints}, where we list the constraint and the conditions (color and orientation) of a single directed edge between $v$ and $w$ or a uni-directed edge pair outgoing at $v$ and incoming at $w$ in $\bar{G}$ under which we have this constraint.

\begin{table}[tb]
 \centering
 \begin{minipage}{0.48\textwidth}
 \begin{tabular}{|c|c|c|c|}
  \hline
  constraint & \ edge \ & \ color \ & \ out \ \\
  \hline
  \multirow{2}{*}{\ $x_1(w) < x_1(v) < x_2(w)$\ }
   & right & $2$ & $v$ \\
   & left & $1$ & $v$ \\
  \hline
  \multirow{2}{*}{$x_1(w) < x_2(v) < x_2(w)$}
  & right & $0$ & $v$ \\
  & left & $3$ & $v$ \\
  \hline
  \multirow{3}{*}{$x_1(w) = x_2(v)$}
  & right & $1$ & $w$ \\
  & left & $2$ & $w$ \\
  & uni & $0$, $3$ & $v$ \\
  \hline
 \end{tabular}
 \end{minipage}
 \hfill
 \begin{minipage}{0.48\textwidth}
 \begin{tabular}{|c|c|c|c|}
  \hline
  constraint & \ edge \ & \ color \  & \ out \ \\
  \hline
  \multirow{2}{*}{\ $y_1(w) < y_1(v) < y_2(w)$\ }
   & right & $3$ & $v$ \\
   & left & $2$ & $v$ \\
  \hline
  \multirow{2}{*}{$y_1(w) < y_2(v) < y_2(w)$}
  & right & $1$ & $v$ \\
  & left & $0$ & $v$ \\
  \hline
  \multirow{3}{*}{$y_1(w) = y_2(v)$}
  & right & $2$ & $w$ \\
  & left & $3$ & $w$ \\
  & uni & $1$, $0$ & $v$ \\
  \hline
 \end{tabular}
 \end{minipage}
 \caption{Constraints encoding the type of contact between $R(v)$ and $R(w)$, defined based on the orientation and color(s) of the edge pair between $v$ and $w$ in $\bar{G}$.
}
 \label{tab:constraints}
\end{table}

We shall show that the system consisting of constraints~\eqref{eq:proper-rectangle} for every vertex $v$ and all constraints in Table~\ref{tab:constraints} is feasible, that is, has a solution.
However, many constraints in this system are implied by other constraints, e.g., the first rows in Table~\ref{tab:constraints} imply~\eqref{eq:proper-rectangle}.
Instead we shall define another set of constraints implying all constraints in Table~\ref{tab:constraints}, for which it is easier to prove feasibility.

First note that no constraint involves both, variables for $x$- and $y$-coordinates.
Hence we can define two independent systems $\mathcal{I}_x$ and $\mathcal{I}_y$ for $x$- and $y$-coordinates, respectively.
In the following we handle only $x$-coordinates.
An analogous argumentation holds for the $y$-coordinates.

In $\mathcal{I}_x$ we have all inequalities of the form~\eqref{eq:proper-rectangle} and all equalities in the left of Table~\ref{tab:constraints}, but only those inequalities in the left of Table~\ref{tab:constraints} that arise from edges in bi-directed edge pairs.
The inequalities in Table~\ref{tab:constraints} arising from uni-directed edge pairs are implied by the following set of inequalities.
For a vertex $v$ in $G$ let $S_1(v) = a_1,\ldots,a_k$ and $S_2(v) = b_1,\ldots,b_\ell$ be the counterclockwise sequences of neighbors of $v$ in the incoming wedges at $v$ bounded by its outgoing edges of color $0$ and $1$, and color $2$ and $3$, respectively.
See the left of Fig.~\ref{fig:define-Ix}.
Then we have in $\mathcal{I}_x$ the inequalities
\begin{equation}
 x_1(a_i) > x_2(a_{i+1}) \text{ for } i = 1,\ldots,k-1 \text{ \ and \ } x_2(b_i) < x_1(b_{i+1}) \text{ for } i=1,\ldots,\ell-1.
\end{equation}
If $k=1$ we have no constraint for $S_1(v)$ and if $\ell = 1$ we have no constraint for $S_2(v)$.

We associate the system $\mathcal{I}_x$ with a partially oriented graph $I_x$ whose vertex set is $\{x_1(v), x_2(v) \mid v \in V\}$.
For each inequality $a > b$ we have an oriented edge $(a,b)$ in $I_x$, while for each equality $a = b$ we have an undirected edge $ab$ in $I_x$.
We refer to Fig.~\ref{fig:define-Ix} for an illustration of this graph.

\begin{figure}[tb]
 \centering
 \includegraphics[scale=0.7]{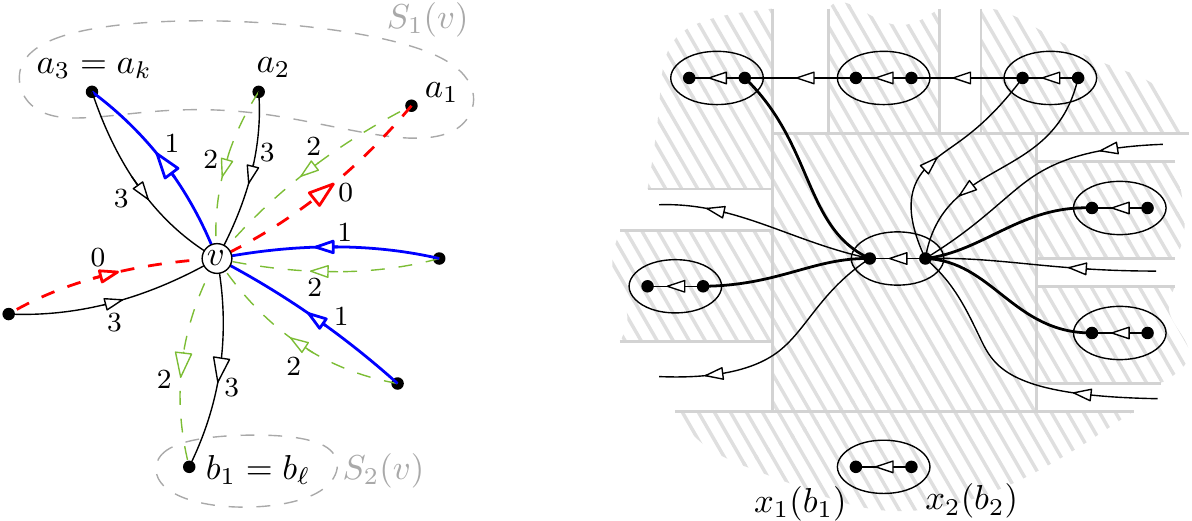}
 \caption{Illustrating the definition of $I_x$ around a vertex $v$. On the right a hypothetical rectangle contact arrangement is indicated.}
 \label{fig:define-Ix}
\end{figure}

\begin{observation}\label{obs:Ix-planar}
 $I_x$ is planar.
\end{observation}
\begin{proof}
 A plane embedding of $I_x$ can be easily defined based on the plane embedding of $G$ as indicated in Fig.~\ref{fig:define-Ix}.
\end{proof}

\begin{observation}\label{obs:continue-path}
 For every inner vertex $u$ of $G$, the vertex $x_1(u)$ in $I_x$ has an incident undirected edge or incident outgoing edge.
 Similarly, $x_2(u)$ has an incident undirected edge or incident incoming edge.
\end{observation}
\begin{proof}
 By symmetry, it is enough to show that $x_1(u)$ has an incident unidirected or outgoing edge.
 Consider the outgoing edge $e$ at $u$ in $\bar{G}$ of color $1$ and its other endpoint $a$.
 If $e$ is a right edge in an edge pair, then $x_1(u) = x_2(a)$ according to Table~\ref{tab:constraints} and hence $x_1(u)x_2(a)$ is an undirected edge in $I_x$.
 If $e$ is a left edge in a bi-directed edge pair, then $x_1(a) < x_1(u)$ according to Table~\ref{tab:constraints} and hence $(x_1(u),x_1(a))$ is a directed edge in $I_x$.
 If $e$ is a left edge in a uni-directed edge pair, then $u$ is contained in the sequence $S_2(a)$, but not the first and not the last in $S_2(a)$.
 Hence $(x_1(u),x_2(b))$ is a directed edge in $I_x$ for the vertex $b$ preceding $u$ in $S_2(a)$.
 Finally, if $e$ is a loose edge, i.e., $a$ corresponds to an inner face of $G$, and $v$ is the neighbor of $u$ in $G$ counterclockwise after $a$, then $u$ is the first vertex in the sequence $S_1(v)$ for vertex $v$ and $|S_1(v)| \geq 2$.
 Hence, $(x_1(u),x_2(b))$ is a directed edge in $I_x$ for the vertex $b$ after $u$ in $S_1(v)$.
\end{proof}

We are now ready to prove Lemma~\ref{lem:labeling-to-rectangles}, i.e., construct a rectangle contact arrangement of $G$ corresponding to the given corner-edge-labeling of $\bar{G}$.

\begin{lemma}\label{lem:labeling-to-rectangles}
 Every corner-edge-labeling of $\bar{G}$ induces a rectangle contact arrangement of $G$.
\end{lemma} 
\begin{proof}
 Consider the variables $\{x_1(v),x_2(v) \mid v \in G\}$, the system $\mathcal{I}_x$ and the corresponding graph $I_x$ as defined above.
 A $u$-to-$v$ path $P$ in $I_x$ is said to be \emph{semi-oriented} if every edge on $P$ is either undirected or directed towards $v$.
 A cycle $C$ in $I_x$ is \emph{semi-oriented} if $C$ seen as a $u$-to-$u$ path is semi-oriented.
 We shall show that $I_x$ has no semi-oriented cycles, which clearly implies that $\mathcal{I}_x$ has a solution.

 First consider a facial cycle $C$ for some inner face of $I_x$.
 Then $C$ is either a triangle with three directed edges (these faces correspond to the inequalities in the first two rows in Table~\ref{tab:constraints} together with the constraint $x_1(w) < x_2(w)$ from~\eqref{eq:proper-rectangle}), or is constituted by the $(x_1(v),x_2(v))$ a directed path $P$ on $S_1(v)$ or $S_2(v)$ and two further edges.
 In the latter case edge $(x_1(v),x_2(v))$ and path $P$ are not directed consistently along $C$, proving that $C$ is not semi-oriented.

 Now assume for the sake of contradiction that $C$ is some semi-oriented cycle in $I_x$ that has $k$ inner faces of $I_x$ in its interior.
 Moreover, assume that $k$ is minimal among all semi-oriented cycles in $I_x$.
 As argued above, $k \geq 2$, i.e., $C$ has some edges or even vertices in its interior.
 We claim that $C$ has a semi-oriented path in its interior with both endpoints in $V(C)$.
 This is clearly the case when $C$ has a chord in its interior.
 Otherwise, we may assume without loss of generality that $x_1(u)$ for some vertex $u$ if $G$ lies in the strict interior of $C$.

 Next we construct a semi-oriented path $P$ starting at $x_1(u)$, adding one edge at a time until we reach a vertex of $C$ or close a semi-oriented cycle.
 Whenever the current end of $P$ is $x_1(u')$ for some $u'$, we continue with the edge given by Obs.~\ref{obs:continue-path}.
 When the current end of $P$ is $x_2(u')$ for some $u'$, we continue with the edge $(x_2(u'),x_1(u'))$.

 If we close a semi-oriented cycle, it surrounds less than $k$ inner faces, contradicting the choice of $C$.
 So we reach $C$ and start to extend $P$ on the other end symmetrically as follows.
 Whenever the current end of $P$ is $x_2(u')$ for some $u'$, we continue with the edge given by Obs.~\ref{obs:continue-path}.
 When the current end of $P$ is $x_1(u')$ for some $u'$, we continue with the edge $(x_2(u'),x_1(u'))$.
 Again, we cannot create a semi-oriented cycle, thus we reach a vertex of $C$ again.
 But now $C \cup P$ consists of two cycles $C_1,C_2$, each surrounding less than $k$ inner faces.
 As $C$ and $P$ are semi-oriented, least one of $C_1$, $C_2$ is semi-oriented as well, contradicting the choice of $C$.

 It follows that $\mathcal{I}_x$ has a solution and analogously we see that $\mathcal{I}_y$ has a solution.
 Defining for every vertex $v$ of $G$ a rectangle $R(v)$ as $[x_1(v),x_2(v)] \times [y_1(v),y_2(v)]$ gives a rectangle contact arrangement corresponding to the given corner-edge-labeling.
\end{proof}

\section{MTP Graphs are Contact Rectangle Graphs -- Proof of Theorem~\ref{thm:MPT-has-orientation}}
\label{sec:MPT-has-orientation}

Let $G$ be an MTP graph.
We shall show that the closure $\bar{G}$ of $G$ admits a $4$-orientation.
Using Thm.~\ref{thm:triangle-free}, namely the bijection between $4$-orientations of $\bar{G}$ and combinatorial rectangle contact arrangements of $G$, this implies that $G$ admits a rectangle contact arrangement.

\begin{backInTimeThm}{thm:MPT-has-orientation}
\begin{theorem}
 Every MTP-graph has a rectangle contact arrangement and it can be computed in linear time.
\end{theorem}
\end{backInTimeThm}
\begin{proof}
 Let $G$ be an MTP-graph.
 We shall show, by induction on the number $n$ of vertices in $G$, that $\bar{G}$ has a $4$-orientation.
 Then by Thm.~\ref{thm:triangle-free} it follows that $G$ has a rectangle contact arrangement, which moreover can be computed in $O(n)$.
 For $n = 4$, i.e., when $G$ is a $4$-cycle, there is nothing to show.
 So assume that $n \geq 5$.
 We distinguish the following cases.

 \medskip

 \textit{Case 1 -- $G$ has an inner face $f$ bounded by a $4$-cycle $C$:}
 Since $G$ is planar and triangle-free, there are two opposite vertices $u,v$ on $C$ that have distance at least $4$ (counted by number of edges) in $G - E(C)$.
 We consider the plane embedded graph $G'$ obtained from $G$ by identifying $u$ and $v$ along $f$ into a new vertex $x$ and removing one edge between $x$ and each of the two vertices $a,b$ in $C - u,v$, see Fig.~\ref{fig:Case1-face}.
 Note that $G'$ is triangle-free by the choice of $u$ and $v$.
 Moreover, there is a natural correspondence between the faces, angles and vertices of $G'$ and those of $G$, except that the face $f$ is missing in $G'$ and $u,v$ are both represented by $x$.

 \begin{figure}[htb]
  \centering
  \includegraphics[width=\textwidth]{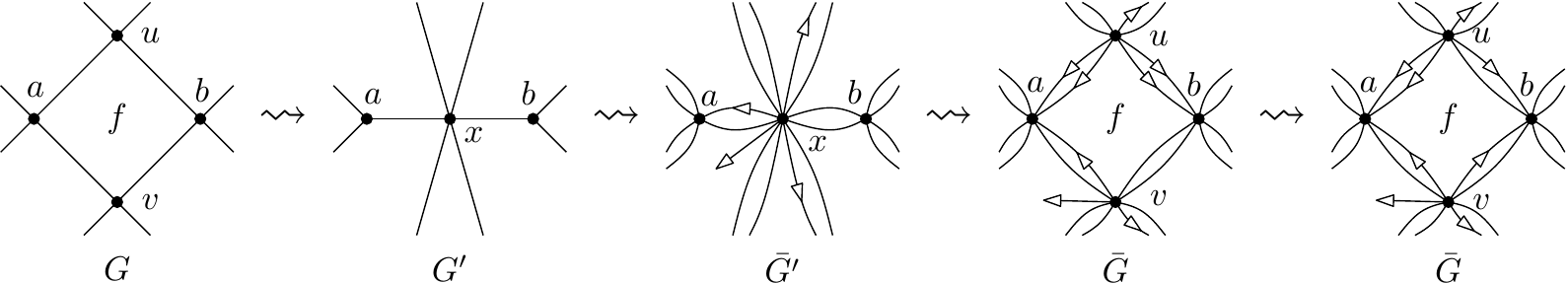}
  \caption{Contracting two opposite vertices at an inner $4$-face of $G$ and obtaining an augmentation with $4$-orientation of $G$ from an augmentation with $4$-orientation of the resulting graph $G'$. Depicting the orientation of edges that are not outgoing at $u$, $v$ or $x$ is omitted.}
  \label{fig:Case1-face}
 \end{figure}

 By induction there is an augmentation $\bar{G'}$ of $G'$ with a $4$-orientation.
 We define an augmentation $\bar{G}$ of $G$ by doubling every edge of $G$ and placing for every an angle of $G$ not in $f$ the same number of pendant edges as in the corresponding angle of $G'$.
 We place no pendant edge in the angles in $f$.
 It is straightforward to check that $\bar{G}$ is indeed an augmentation of $G$.
 Secondly, we define a $4$-orientation of $\bar{G}$ based on the given $4$-orientation of $\bar{G'}$.
 Every edge not in $C$ is oriented the way it is in $\bar{G'}$, every edge at $v$ in $C$ is oriented like the corresponding edges at $x$ in $\bar{G'}$, and every edge at $u$ in $C$ is oriented outgoing at $u$.
 Now, $u$ has outdegree $4+i$ and $v$ has outdegree $4-i$ for some $i \in \{0,1,2,3,4\}$.
 In particular, $v$ has at least $i$ incoming edges in $C$.
 Reversing $i$ edge-disjoint paths of the form $u,w,v$, where $w \in \{a,b\}$, gives a $4$-orientation of $\bar{G}$, see again Fig.~\ref{fig:Case1-face}.

 \medskip

 \textit{Case 2 -- Every inner face of $G$ has length at least $5$:}
 If $G$ has no $4$-cycle at all, consider an arbitrary edge $e = uv$.
 Otherwise, $G$ has a number of $4$-cycles with non-empty interior, which are partially ordered by inclusion of their interior.
 In this case consider an edge $e = uv$ of $G$ strictly in the interior of a $4$-cycle with inclusion-minimal interior.
 In any case, $e$ lies on no $4$-cycle and hence $u$ and $v$ have distance at least $4$ (counted by number of edges) in $G - e$.
 We consider the plane embedded graph $G'$ obtained from $G$ by contracting $e$, i.e., identifying $u$ and $v$ along $e$ into a new vertex $x$, see Fig.~\ref{fig:Case2-edge}.

 \begin{figure}[htb]
  \centering
  \includegraphics{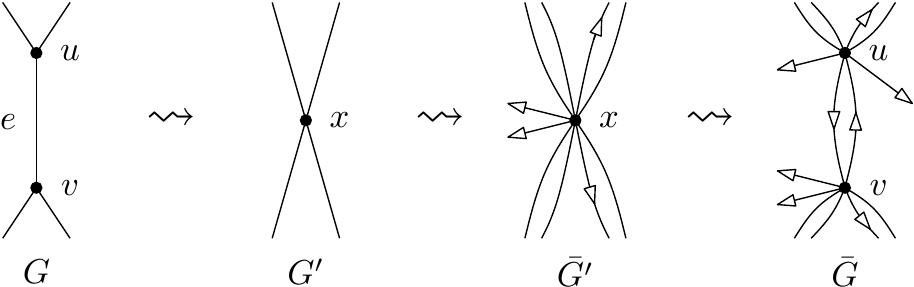}
  \caption{Contracting an edge of $G$ and obtaining an augmentation with $4$-orientation of $G$ from an augmentation with $4$-orientation of the resulting graph $G'$. Depicting the orientation of edges that are not outgoing at $u$, $v$ or $x$ is omitted.}
  \label{fig:Case2-edge}
 \end{figure}

 Similarly to the Case~1 we have by induction an augmentation $\bar{G'}$ of $G'$ with a $4$-orientation.
 We uncontract $x$ back into $u$ and $v$ with two parallel edges in between and place all (if any) pendant edges at $x$ to $v$.
 Without loss of generality assume that $u$ has $i \leq 2$ outgoing edges.
 Into each of the two angles bounded by $e$ we add one pendant edge at $u$.
 Again, it is straightforward to check that this defines an augmentation $\bar{G}$ of $G$.
 We obtain a $4$-orientation of $\bar{G}$ by orienting exactly $i-2$ of two parallel edges between $u$ and $v$ outgoing at $u$, see again Fig.~\ref{fig:Case2-edge}.
\end{proof}

\section{Line-Pierced Rectangle Arrangements and Squarability --\\ Proof of Theorem~\ref{thm:line-pierced}}
\label{sec:line-pierced}

Recall that a rectangle arrangement~$\mathcal R$ is line-pierced if there is a horizontal line $\ell$ that intersects every rectangle in $\mathcal R$.
Note that by the line-piercing property of $\mathcal R$ the intersection graph remains the same if we project each rectangle $R=[a,b] \times [c,d] \in \mathcal R$ onto the interval $[a,b] \subseteq \mathbb R$.
In particular, the intersection graph $G_{\mathcal R}$ of a line-pierced rectangle arrangement $\mathcal R$ is an \emph{interval graph}, i.e., intersection graph of intervals on the real line.

Line-pierced rectangle arrangements, however, carry more information than one-dimensional interval graphs since the vertical positions of intersection points between rectangles do influence the combinatorial properties of the arrangement.
We obtain two squarability results for line-pierced arrangements in Propositions~\ref{pro:lp-tf-squared} and~\ref{pro:lp-unit-squares}, which yield Thm.~\ref{thm:line-pierced}.

\wormholeProp{pro:lp-tf-squared}
\begin{proposition}\label{pro:lp-tf-squared}
	Every line-pierced, triangle-free, and cross-free rectangle arrangement $\mathcal R$ is squarable.
\end{proposition}

\begin{proof}
	First, we observe that the intersection graph $G_{\mathcal R}$ of $\mathcal R$ is a caterpillar, i.e., a tree for which a \emph{backbone path} remains after removing all leaves.
	As observed above, $G_{\mathcal R}$ is an interval graph and triangle-free interval graphs are known to be caterpillars~\cite{e-eig-93}.

	Whether a rectangle $R=[a,b] \times [c,d]$ forms a vertex of the caterpillar's backbone path $P$ or a leaf depends on the intersection patterns of its corresponding interval $[a,b]$ with the other intervals.
	If we have $[a,b] \subseteq [a',b']$ for another rectangle $R' = [a',b'] \times [c',d']$ then $R$ must be a leaf in $G_{\mathcal R}$. Otherwise $[a,b]$ (and thus also $[a',b']$) would need to intersect a third interval and form a triangle.
	The leftmost and rightmost rectangles of $\mathcal R$ may also be leaves, but for any other rectangles $R$ it holds true that there are a left neighbor $R_l = [a_l,b_l] \times [c_l,d_l]$ and a right neighbor $R_r = [a_r,b_r] \times [c_r,d_r]$ such that $a_l < a < b_l < a_r < b < b_r$ and thus $R$ forms a vertex of $P$.

	Because $\mathcal R$ is cross-free, it is clear that a leaf $R$ of $G_{\mathcal R}$ must be a side-piercing rectangle of either the top or the bottom side of its neighboring rectangle $R'$ or be fully contained in $R'$.
	 Two neighboring rectangles of $P$ may be side-piercing on their left or right sides or form a corner intersection.

	We now construct a squaring $\mathcal R'$ of $\mathcal R$.
	It is obvious that the rectangles of the backbone path $P$ (including the leftmost and rightmost rectangles of $\mathcal R$, even if they are leaves) can be realized as combinatorially equivalent squares with horizontal overlaps of $\varepsilon>0$, see Fig.~\ref{fig:caterpillar}.
	We denote the square representing a rectangle $R$ as $S_R$.
	Consider a rectangle $R$ of $P$ and its leaves in $G_{\mathcal R}$.
	Because $\mathcal R$ is triangle-free, we can order the leaves from left to right by their x-intervals.
	We represent all leaves of $R$ by squares of the same side length, which is chosen such that their sum is less than the side length of $S_R$ minus $2 \varepsilon$.
	It is again easy to see that these squares can be placed in the prescribed order and combinatorially equivalent to $\mathcal R$ as illustrated in Fig.~\ref{fig:caterpillar}; hence $\mathcal R$ is indeed squarable.
\end{proof}

\begin{figure}[tb]
	\centering
	\includegraphics{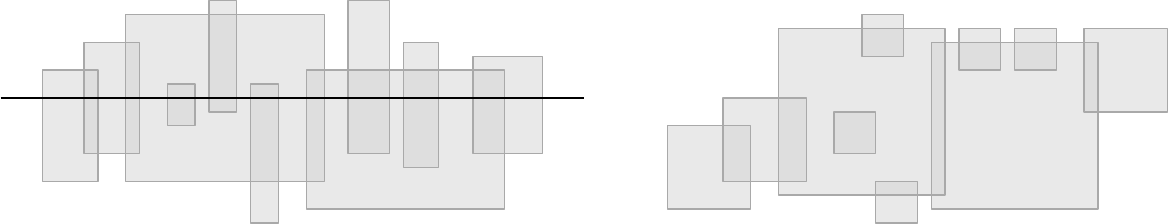}
	\caption{Constructing a combinatorially equivalent squaring from a line-pierce, triangle-free, and cross-free rectangle arrangement.}
	\label{fig:caterpillar}
\end{figure}

There are instances, however, that satisfy the conditions of Prop.~\ref{pro:lp-tf-squared} and thus have a squaring, but not a line-pierced one.
An example is given in Fig.~\ref{fig:no-lp-squaring}.
Assume the square representing $R$ has side length $1$.
In the projection of any square arrangement to the y-axis, the two intervals of $S$ and $T$ must overlap (due to being line-pierced) and their union contains the interval of $R$.
However, when projected to the x-axis, the projections of $S$ and $T$ must be disjoint and they must be contained in the projection of $R$.
This is impossible for squares.

\begin{figure}[tb]
	\centering
	\includegraphics{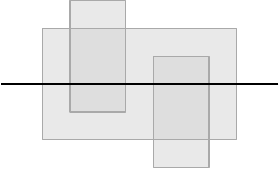}
        \hspace{1.5em}
	\includegraphics{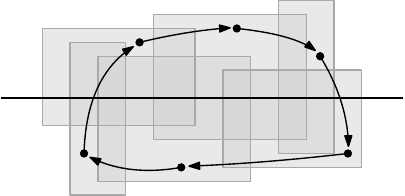}
        \hspace{1.5em}
	\includegraphics{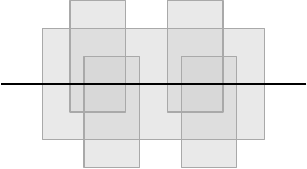}
	\caption{Left: A line-pierced, triangle-free rectangle arrangement that has no line-pierced squaring. Middle: An unsquarable line-pierced rectangle arrangement due to a forbidden cycle of side-piercing intersections. Right: Squaring the two vertical pairs of rectangles on the right implies that the central square would need to be wider than tall.}
	\label{fig:no-lp-squaring}
\end{figure}

Requiring that the line-pierced arrangement $\mathcal R$ is triangle- and cross-free thus is sufficient for its squarability but we cannot guarantee that a line-pierced squaring exists.
Another condition that we have used before is the restriction of the intersection types to corner intersections.
Regardless of the existence of triangles in the intersection graph this allows us to prove an even stronger result.

\begin{proposition}\label{pro:lp-unit-squares}
	Every line-pierced rectangle arrangement $\mathcal R$ restricted to corner intersections is squarable. 
	There even exists a corresponding squaring with unit squares that remains line-pierced.
\end{proposition}
\begin{proof}
	From the restriction to corner intersections we can derive certain properties on the x- and y-orders of the boundaries of two intersecting rectangles.
	Let $R=[a,b] \times [c,d]$ and $R'=[a',b'] \times [c',d']$ be two corner-intersecting rectangles and without loss of generality assume $a<a'$.
	This implies that the order of the left and right sides of $R$ and $R'$ is $a<a'<b<b'$.
	In particular, the horizontal order of the left sides and the horizontal order of the right sides of the rectangles in $\mathcal R$ imply the same total order $\prec$ of $\mathcal R$.\footnote{It is known that this implies that $G$ is a unit-interval graph. However, it is instrumental to repeat the proof here.}
	For the vertical order of two intersecting rectangles there are the two alternatives $c<c'<d<d'$ ($R$ is to the lower left of $R'$) or $c'<c<d'<d$ ($R$ is to the upper left of $R'$) and the horizontal piercing line $\ell$ is placed between $\max \{c,c'\}$ and $\min \{d,d'\}$.

	We now use induction to construct a line-pierced squaring of $\mathcal R$ with unit squares.
	It is obvious that the statement of the proposition holds true for $|\mathcal R| = 1$.
	So let $n>1$ and consider the arrangement $\mathcal R' = \mathcal R \setminus \{L\}$, where $L = [a,b] \times [c,d]$ is the leftmost rectangle in $\mathcal R$.
	Since $|\mathcal R'| = n-1$ a line-pierced squaring $\mathcal S'$ with unit squares exists.
	We add the unit square $L'$ representing $L$ to $\mathcal S'$ by first determining its horizontal position and then its vertical position.

	Let $A_1, \dots, A_k$ be the rectangles intersecting $L$ in $\mathcal R$,  ordered from left to right according to the total order $\prec$, where each rectangle $A_i$ is specified as $A_i = [a_i,b_i] \times [c_i,d_i]$.
	Since $\mathcal R$ is line-pierced, $L$ and $A_1, \dots, A_k$ form a clique and we know that $a < a_1 < \dots < a_k < b < b_1 < \dots < b_k$.
	By the induction hypothesis, the squaring $\mathcal S'$ satisfies that (i) the coordinates of the corresponding squares $A'_1, \dots, A'_k$ with $A'_i = [a'_i,b'_i] \times [c'_i,d'_i]$ are ordered as $a'_1 < \dots < a'_k <  b'_1 < \dots < b'_k$ and that (ii) $b'_i-a'_i = 1$ for each $i=1, \dots, k$.
	We set the right x-coordinate $b'$ of $L'$ such that $a'_k < b' < b'_1$ and the left x-coordinate $a' = b' - 1$.
	Since $a'_1 = b'_1 - 1$ and $b' < b'_1$ it follows that $a' < a'_1$.
	This already implies that the intersection graphs of $\mathcal R$ and $\mathcal S = \mathcal S' \cup \{L'\}$ are identical as long as $\mathcal S$ is line-pierced and thus intersects all $A'_1, \dots A'_k$.

	It remains to determine the y-coordinates $c'$ and $d'$ of $L'$.
	Since $L$ and its neighbors form a clique in $\mathcal R$ they can again be totally ordered by either their lower or their upper y-coordinates and the resulting two orders coincide due to the fact that we have only corner intersections.
	In particular, there are two rectangles $A_i$ and $A_j$ ($1 \le i,j \le k$) that are the immediate neighbors of $L$ in that order, where $c_i < c < c_j < d_i < d < d_j$ (unless $L$ is top- or bottommost).
	Again, by the induction hypothesis, the vertical order of the unit squares $A'_1, \dots A'_k$ in $\mathcal S'$ is identical to the vertical order of $A_1, \dots, A_k$ in $\mathcal R$.
	Furthermore, $d'_l - c'_l = 1$ for $l=1, \dots, k$ and hence we can select $[c',d']$ such that $d'-c'=1$, the y-coordinates satisfy $c'_i < c' < c'_j <  d'_i < d' < d'_j$, and $\ell$ lies between $c'$ and $d'$.
	This concludes the proof and also immediately yields a linear-time construction algorithm for a sorted input arrangement~$\mathcal R$. %
\end{proof}

If we drop the restrictions to corner intersections and triangle-free arrangements, we can immediately find unsquarable instances, either by creating cyclic "`smaller than"' relations or by introducing intersection patterns that become geometrically infeasible for squares.
Two examples are given in Fig.~\ref{fig:no-lp-squaring}. 

\section{Conclusions}
\label{sec:conclusions}

We have introduced corner-edge-labelings, a new combinatorial structure similar to Schnyder realizers, which captures the combinatorially equivalent maximal rectangle arrangements with no three rectangles sharing a point.
Using this, we proved that every triangle-free planar graph is a rectangle contact graph.
We also introduced the squarability problem, which asks for a given rectangle arrangement whether there is a combinatorially equivalent arrangement using only squares.
We provide some forbidden configuration for the squarability of an arrangement and show that certain subclasses of line-pierced arrangements are always squarable.
It remains open whether the decision problem for general arrangements is NP-complete.

Surprisingly, every unsquarable arrangement that we know has a crossing or a side-piercing.
Hence we would like to ask whether every rectangle arrangement with only corner intersections is squarable.
Another natural question is whether every triangle-free planar graph is a square contact graph.

\end{document}